\documentclass[letterpaper, 10 pt, conference]{ieeeconf}  % Comment this line out if you need a4paper
\usepackage{algorithm}
\usepackage{algorithmicx}
\usepackage{stfloats}
\usepackage{graphicx}
\graphicspath{{Figures/}}          % Include this line if your 
\usepackage{mathrsfs}         
\usepackage{amsmath}   
\usepackage{amsfonts}    
\usepackage{amssymb}   
\usepackage{booktabs}
\usepackage{threeparttable}
\usepackage{subfig}		
\usepackage{float}		
\usepackage{tabularx}
\usepackage{comment}
\usepackage{cite}
\usepackage[dvipsnames]{xcolor}
\usepackage{color}

\usepackage[font=small,skip=0pt]{caption}

\usepackage{enumerate}

\usepackage{amsthm}
\allowdisplaybreaks[4]
\newtheorem{remark}{\textbf{Remark}}

\newtheorem{assumption}{\emph{\textbf{Assumption}}}
\newtheorem{theorem}{\textbf{Theorem}}
\newtheorem{lemma}{\textbf{Lemma}}

\newtheorem{example}{Example}
\IEEEoverridecommandlockouts                   
\title{\LARGE \bf
	Distributed Cooperative Multi-Agent Reinforcement Learning with Directed Coordination Graph
}

\author{Gangshan Jing,~He Bai,~Jemin George,~Aranya Chakrabortty~and~Piyush. K. Sharma% <-this % stops a space
	\thanks{G.~Jing and A. Chakrabortty are with  North Carolina State University, Raleigh, NC 27695, USA.
		{\tt\small \{gjing, achakra2\}@ncsu.edu}}%
	\thanks{H.~Bai is with Oklahoma State University, Stillwater, OK 74078, USA.
		{\tt\small he.bai@okstate.edu}}%
	\thanks{J.~George and P.~Sharma are with the DEVCOM U.S. Army Research Laboratory, Adelphi, MD 20783, USA.
		{\tt\small \{jemin.george.civ,piyush.k.sharma.civ\}@army.mil}}%
}

\begin{document}
	
	\setlength{\abovedisplayskip}{4pt}
	\setlength{\belowdisplayskip}{4pt}
	
	\maketitle
	%\thispagestyle{empty}
	%\pagestyle{empty}

	%%%%%%%%%%%%%%%%%%%%%%%%%%%%%%%%%%%%%%%%%%%%%%%%%%%%%%%%%%%%%%%%%%%%%%%%%%%%%%%%
	\begin{abstract}
	Existing distributed cooperative multi-agent reinforcement learning (MARL) frameworks usually assume undirected coordination graphs and communication graphs while estimating a global reward via consensus algorithms for policy evaluation. Such a framework may induce expensive communication costs and exhibit poor scalability due to requirement of global consensus. In this work, we study MARLs with directed coordination graphs, and propose a distributed RL algorithm where the local policy evaluations are based on local value functions. The local value function of each agent is obtained by local communication with its neighbors through a directed learning-induced communication graph, without using any consensus algorithm. A zeroth-order optimization (ZOO) approach based on parameter perturbation is employed to achieve gradient estimation. By comparing with existing ZOO-based RL algorithms, we show that our proposed distributed RL algorithm guarantees high scalability. A distributed resource allocation example is shown to illustrate the effectiveness of our algorithm.
	\end{abstract}
	

	\section{Introduction}
Traditional Reinforcement Learning (RL) aims to find an optimal policy for a single agent to accomplish a specific task by making the agent interact with the environment. The problem, however, is much more complex for multi-agent reinforcement learning (MARL) due to the non-stationary environment for each agent and the curse of dimensionality. MARL, therefore, has been attracting increasing attention recently, and has been studied extensively in e.g., \cite{kar2013cal,zhang2018fully,zhang2020cooperative,chen2021communication,cassano2020multiagent,gronauer2021multi,qu2020scalable}. Markov game \cite{littman1994markov} has been commonly employed as a mathematical formulation for MARL, where each agent has its own reward and different agents may have conflicts of interest. In distributed cooperative MARL, e.g., \cite{kar2013cal,zhang2018fully,zhang2020cooperative}, the global reward is usually formulated as the average of the sum of all agents' local rewards. Accordingly, policy gradient algorithms always require each agent to track the global reward by interacting with its neighbors via a connected undirected communication graph. Moreover, due to the non-stationarity of the environment, the global state information is assumed to be known for all agents in several references, e.g.,  \cite{zhang2018fully,chen2021communication}. In \cite{qu2020scalable}, the authors considered a more distributed setting, where each agent has its local state space and local action space, and different agents are coupled through network-based interactions. This is similar to the setting in the present work.
	
In this paper, we study distributed cooperative MARL with partial observations and a directed coordination graph. We consider a general case where the coordination graph and the communication graph are different but related to each other. The coordination graph is considered to be weakly connected, which specifies the coordination relationship between agents and determines the partial observation of each agent. The communication graph is designed for each agent to obtain the value of its \emph{local value-function}, which is a directed graph dictated by the coordination graph. We show that under an appropriate design for local value-functions, the local gradient of the global value-function with respect to the policy of each agent is equivalent to the local gradient of a local value-function (see Lemma \ref{le lg=gg}). As a result, gradient estimation can be achieved without knowing the global value. 
	
The main contributions of our work in contrast to the existing literature are summarized as follows:
	
	\begin{enumerate}	
	\item We employ a weakly connected graph to characterize the coordination relationships among agents, which also determines the partial observation of each agent. The notion ``coordination graph" has been used in \cite{guestrin2002coordinated,kok2006collaborative} to describe limited communication and observability in the network and enable distributed Q-learning. However, this graph was always considered bidirectional in those works, while we use a general directed graph. 
	
	\item We establish a connection between the coordination graph and the communication graph required for learning, based on which each agent can estimate its local gradient via communications with its designated set of agents. The agents are not required to reach consensus on any global index. Such a setting is more scalable to large-scale networks compared with the distributed RL based on average consensus, e.g., \cite{kar2013cal,zhang2018fully,zhang2020cooperative}.
	
	\item Distributed MARL via value-functions involving only partial agents has also been studied in \cite{qu2020scalable}, where the local value-function of each agent is designed to approximate the global value-function, by ignoring some other agents. In our work, there is no error induced by the use of local value-functions in the sense that the partial gradient of each local value-function is equivalent to that of the global value-function; see Lemma \ref{le lg=gg}.
	
	\item Our distributed RL algorithm is a policy gradient algorithm based on zeroth-order optimization (ZOO) with exploration in the parameter space. It is known that the dimension of the parameter space is typically higher than that of the action space \cite{vemula2019contrasting,kumar2020zeroth}, which is the main shortage of parameter-based ZOO compared with action-based ZOO for RL. In this paper, the sample space dimension for each agent is significantly reduced due to the use of local value-functions. Moreover, under different zeroth-order oracles, we show that our learning framework always exhibits a reduced variance of the estimated gradient compared with global value-based policy evaluation. Note that most of the existing distributed MARL works (e.g., \cite{zhang2018fully,zhang2020cooperative,cassano2020multiagent}) require estimating the value of the global value-function.	
\end{enumerate}

%This work can be viewed  as the synchronous version of our previous work \cite{jing2021asynchronous}. Besides the different updating scheme, the problem formulation in this paper is more general, and a different sample space is adopted to realize synchronous sampling and updating. Moreover, similar to \cite{jing2021asynchronous}, the proposed distributed algorithm in this work is applicable to the distributed linear quadratic regulator design problem.

The rest of this paper is organized as follows. Section \ref{sec PF} describes the MARL formulation and the main goal of this paper. Section \ref{sec local} proposes  the local value-function design and specifies the communication graph required for learning. Section \ref{sec algorithm} shows the distributed RL algorithm, and provides convergence and variance analysis. Section \ref{sec sim} shows the simulation result of an application example. Section \ref{sec: conclusion} concludes the paper.

	\textbf{Notation}: Throughout the paper, $\mathcal{G}=(\mathcal{V},\mathcal{E})$ is an unweighted directed graph, where $\mathcal{V}=\{1,...,N\}$ is the set of vertices, $\mathcal{E}\subset\mathcal{V}\times\mathcal{V}$ is the set of edges, $(i,j)\in\mathcal{E}$ means that there is a directional edge in $\mathcal{G}$ from $i$ to $j$. A path from $i$ to $j$ is a sequence of distinct edges of the form $(i_1,i_2)$, $(i_2,i_3)$, ..., $(i_{r-1},i_r)$ where $i_1=i$ and $i_r=j$. We use $i\stackrel{\mathcal{E}}{\longrightarrow} j$ to denote that there is a path from $i$ to $j$ in edge set $\mathcal{E}$. A subgraph $\mathcal{G}'=(\mathcal{V}',\mathcal{E}')$ with $\mathcal{V}'\subseteq\mathcal{V}$ and $\mathcal{E}'\subseteq\mathcal{E}$ is said to be a strongly connected component (SCC) if there is a path between any two vertices in $\mathcal{G}'$. One vertex is a special SCC. The $d\times d$ identity matrix is denoted by $I_d$, the $a\times b$ zero matrix is denoted by $\mathbf{0}_{a\times b}$. $\mathbb{R}^d$ is the $d$-dimensional Euclidean space. $\mathbb{N}$ is the set of non-negative integers.  %The Kronecker product is denoted by $\otimes$. Given a matrix $X$, $X\succeq0$ implies that $X$ is positive semi-definite;  $\im(X)$ denotes the image space of $X$. We use $\diag\{A_1,...,A_N\}$ to denote a block diagonal matrix with $A_i$'s on the diagonal. Matrix $G_{ab}(s)$ denotes the transfer function from input $a$ to output $b$. Given a transfer function $G(s)$, its $\mathcal{H}_\infty$ norm is $\|G(s)\|_\infty=\sup_\omega\sigma_{\max}(G(j\omega))$, its $\mathcal{H}_2$ norm is $\|G(s)\|_2=(\frac{1}{2\pi}\int_{-\infty}^{\infty}\tr(G^\top(j\omega)G(j\omega))d\omega)^{1/2}$. The Euclidean norm is denoted by $\|\cdot\|$.

	\section{Problem Formulation}\label{sec PF}
	
	\subsection{Cooperative MARL with Partial Observation}
	
	We model the MARL as a tuple $\mathcal{M}=(\mathcal{V},\mathcal{E}_C, \Pi_{i\in\mathcal{V}}\mathcal{M}_i, \Pi_{i\in\mathcal{V}}\mathcal{O}_i,\Pi_{i\in\mathcal{V}}R_i,\gamma)$ with $\mathcal{M}_i=(\mathcal{S}_i, \mathcal{A}_i,\mathcal{P}_i)$ denoting the process for agent $i$, where 
	\begin{itemize}
		\item $\mathcal{V}=\{1,...,N\}$ is the set of agent indices;
		
		\item $\mathcal{E}_C\subseteq\mathcal{V}\times\mathcal{V}$ is the edge set of the {\it coordination graph} $\mathcal{G}_C=(\mathcal{V},\mathcal{E}_C)$, which specifies the coordination relationship between agents;
		
		\item $\mathcal{S}_i$ and $\mathcal{A}_i$ are the state space and the action space of agent $i$, respectively;
		
		\item $\mathcal{P}_i:\Pi_{j\in\mathcal{I}_i}\mathcal{S}_j\times \Pi_{j\in\mathcal{I}_i}\mathcal{A}_j\times\mathcal{S}_i\rightarrow P(\mathcal{S}_i)$ is the transition probability function specifying the probability of transition to state $s'_i\in\mathcal{S}_i$ under states $\{s_j\}_{j\in\mathcal{I}_i}$ and actions $\{a_j\}_{j\in\mathcal{I}_i}$, here $P(\mathcal{S}_i)$ is the set of probability measures on $\mathcal{S}_i$, $\mathcal{I}_i=\{j\in\mathcal{V}:(j,i)\in\mathcal{E}_C\}\cup\{i\}$;
		
		\item $R_i:\mathcal{S}_i\times\mathcal{A}_i\rightarrow\mathbb{R}$ is the immediate reward returned to agent $i$ when agent $i$ takes action $a_i\in\mathcal{A}_i$ at the current state $s_i\in\mathcal{S}_i$;
		
		\item $\mathcal{O}_i=\Pi_{j\in\mathcal{I}_i}\mathcal{S}_j$ is the observation space\footnote{The proposed results in this work are also applicable to cases when the observation space of agent $i$ is only a subset of $\Pi_{j\in\mathcal{I}_i}\mathcal{S}_j$, e.g., $\mathcal{S}_i$.} of agent $i$, which includes the states of all the agents in $\mathcal{I}_i$;
		
		\item $\gamma\in(0,1]$ is the discount factor that trades off the instantaneous and future rewards.
	\end{itemize}
	
	Let $\mathcal{S}=\Pi_{j\in\mathcal{V}}\mathcal{S}_j$, $\mathcal{A}=\Pi_{j\in\mathcal{V}}\mathcal{A}_j$, and $\mathcal{P}=\Pi_{j\in\mathcal{V}}\mathcal{P}_j$ denote  the state space, action space and state transition matrix of the whole multi-agent system. Let $\pi:\mathcal{S}\times \mathcal{A}\rightarrow P(\mathcal{A})$ and $\pi_i: \mathcal{O}_i\times\mathcal{A}_i\rightarrow P(\mathcal{A}_i)$ be the global policy function of the MAS and the local policy function of agent $i$, respectively. Here $P(\mathcal{A})$ and $P(\mathcal{A}_i)$ are the sets of probability measures on $\mathcal{A}$ and $\mathcal{A}_i$, respectively. The global policy is the policy of the whole MAS from the centralized perspective, thus is based on the global state $s$. The local policy of agent $i$ is the policy from agent $i$'s point of view, which is based on $o_i$ composed of partial states. Note that a global policy always corresponds to a collection of local policies uniquely.
	
	At each time step $t$ of the MDP, each agent $i\in\mathcal{V}$ executes an action $a_{i,t}\in\mathcal{A}_i$ according to its policy and the local observation $o_{i,t}\in\mathcal{O}_i$, resulting in transition to a new state $s_{i,t+1}$, then obtains a reward $r_{i,t+1}(s_{i,t},a_{i,t})$.  Note that such a formulation is different from that in \cite{zhang2018fully,chen2021communication}, where the transition and reward of each agent are associated with the global state $s_t$.

	%From the definition, each agent has its independent state space, action space, and transition probability. Given local policy $\pi_i$, each agent takes its action based on the local observation $o_i$, which is independent of  its neighbors' states, in contrast to the global observation assumption in \cite{zhang2018fully,chen2021communication}.  Moreover, different agents are coupled in their local rewards, meaning that the reward returned to each agent depends on not only its own action, but also its collaborative neighbors' actions. Intuitively, the assumption on local reward here is more restrictive than the independent local reward assumption in \cite{zhang2018fully}, it is, however, valid in many application scenarios. For example, a group of robots aim to cooperatively move supplies from other locations to location $A$. The amount of supplies at location $A$ can be the common reward for all the robots, and it certainly depends on all the robots' actions.

	The long-term accumulated discounted global reward is defined as
	\begin{equation}
		R(s,a)=\sum_{t=0}^{T-1}\gamma^tr_t(s_t,a_t)=\sum_{i=1}^NR_i(s_i,a_i),
	\end{equation} 
where $R_i(s_i,a_i)=\sum_{t=0}^{T-1}\gamma^tr_{i,t}(s_{i,t},a_{i,t})$, $r_t$ is the global reward for the MAS at time $t$, $R_i(s,a)$ is the long-term individual reward for agent $i$, $r_{i,t}$ is the individual reward for agent $i$ at  time $t$, and $T$ is the number of iterations we consider. Note that maximizing $R(s,a)$ is equivalent to maximizing its average $\frac{1}{N}R(s,a)$, which has been commonly adopted as the learning objective in many MARL references, e.g. \cite{kar2013cal,zhang2018fully,zhang2020cooperative,qu2020scalable}. 

Based on the long term global reward, with a given policy $\pi$, we are able to define the global state value-function $V^\pi(s')=\mathbb{E}[R(s,\pi(s))|s_0=s']$ and state-action value-function $Q^\pi(s',a')=\mathbb{E}[R(s,a)|s_0=s',a_0=a']$, which describe the expected long term global reward when agents have initial state $s'$ and initial state-action pair $(s',a')$, respectively. Similarly, a local state value-function with initial state $s'$ for each agent $i$ can be defined as $V_i^\pi(s')=\mathbb{E}[R_i(s_i,a_i)|s_0=s']$.
	
The goal of the MARL problem in this paper is to find a solution to
	\begin{equation}
		\begin{split}
			\max_{\pi} J(\pi):=\mathbb{E}_{s_0\sim\mathcal{D}}V^{\pi}(s_0),
		\end{split}
	\end{equation}
	where $\mathcal{D}$ denotes the distribution that the initial state follows. Our approach is also applicable to other formulations for $J(\pi)$, for example, $J(\pi)=\mathbb{E}_{s_0\sim\mathcal{D}_s,a_0\sim\mathcal{D}_a}Q^\pi(s_0,a_0)$, where $\mathcal{D}_s$ and $\mathcal{D}_a$ are the two distributions that $s_0$ and $a_0$ follow, respectively. For convenience of analysis, we also define the expected value to be maximized corresponding to individual reward for each agent $i$ as 
	\begin{equation}
		J_i(\pi)=\mathbb{E}_{s_0\sim\mathcal{D}}V_i^\pi(s_0), i\in\mathcal{V}.
	\end{equation}

	To solve for an optimal policy such that $J(\pi)$ can be maximized, without loss of generality, we parameterize the global policy $\pi(s,a)$ using parameters $\theta=(\theta_1^\top,...,\theta_N^\top)^\top\in\mathbb{R}^d$ with $\theta_i\in\mathbb{R}^{d_i}$. The global policy and agent $i$'s local policy are then rewritten as $\pi^\theta(s,a)$ and $\pi^{\theta_i}_i(o_i,a_i)$, respectively. As a result, it suffices to solve the following optimization problem:
	\begin{equation}\label{maxJ}
		\max_{\theta} J(\theta):=\mathbb{E}_{s_0\sim\mathcal{D}}V^{\pi(\theta)}(s_0).
	\end{equation}

    A typical method for solving (\ref{maxJ}) is implementing the policy gradient algorithm for each agent $i$:
    \begin{equation}\label{pg}
    	\theta^{k+1}_i=\theta^k_i+\eta\nabla_{\theta_i} J(\theta^k), i\in\mathcal{V},
    \end{equation} 
    where $\eta>0$ is the step-size.    
    
    Existing distributed policy gradient methods such as actor-critic~\cite{zhang2018fully}  and zeroth-order optimization \cite{zhang2020cooperative} can always be employed to estimate $\nabla_{\theta_i}J(\theta^k)$ when there is a connected undirected communication graph among the agents. However, these approaches are based on estimating the value of the global value-function, which requires a large amount of communications during each learning episode. Moreover, policy evaluation based on the global value-function has a significant scalability issue due to the high dimension of the state and action spaces for large-scale networks. 
        
    In the next section, by utilizing the structure of the directed coordination graph $\mathcal{G}_C$, we will design a local value-function for each agent, based on which the local gradient can be derived with only partial agents involved.

\section{Local Value-Function and Communication Graph Design}\label{sec local}

In this section, we introduce a local value-function for each agent and propose a graph condition for inter-agent communications required for distributed learning.

\subsection{Graph Clustering and Local Value-Function Design}

We consider the coordination graph $\mathcal{G}_C$ as a weakly connected directed graph\footnote{When $\mathcal{G}_C$ is not weakly connected, the problem can be naturally decomposed to independent smaller-sized MARL problems.}. Note that SCCs can always be found in $\mathcal{G}_C$. Suppose that we already have a clustering $\mathcal{V}=\cup_{j=1}^n\mathcal{V}_j$, where $\mathcal{V}_j$ is the vertex set of the $j$-th maximum independent SCC in $\mathcal{G}_C$, and $\mathcal{V}_j\cap\mathcal{V}_k=\varnothing$ for any distinct $j,k\in\{1,...,n\}$.
\begin{figure}
	\centering
	\includegraphics[width=8cm]{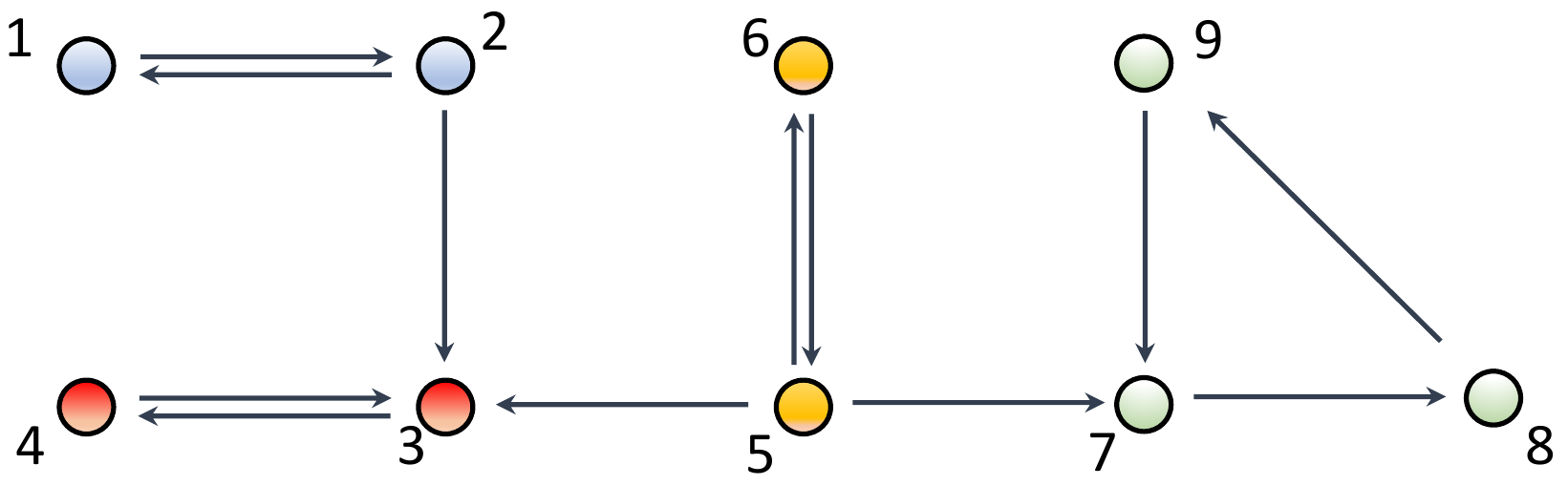}
	\caption{A coordination graph with 9 agents and 4 independent maximum SCCs.} \label{fig coordination graph}
\end{figure}
An example for the coordination graph is shown in Fig. \ref{fig coordination graph}, where there are 4 clusters corresponding to 4 maximum independent SCC, i.e., $\mathcal{V}_1=\{1,2\}$, $\mathcal{V}_2=\{3,4\}$, $\mathcal{V}_3=\{5,6\}$, and $\mathcal{V}_4=\{7,8,9\}$.

Let $$\mathcal{N}^L_i=\{j\in\mathcal{V}: i\stackrel{\mathcal{E}}{\longrightarrow} j\}$$ be the set of vertices in graph $\mathcal{G}$ that are reachable from agent $i$. Denote $\mathcal{I}^L_i=\mathcal{N}^L_i\cup\{i\}$. It is observed that $\mathcal{I}_i^L=\mathcal{I}_j^L$ if $i$ and $j$ belong to the same cluster. 

Define the following composite reward for agent $i$:
\begin{equation}
	\hat{r}_{i,t}=\sum_{j\in\mathcal{I}_i^L}r_{j,t}(s_{j,t},a_{j,t}).
\end{equation}
Accordingly, we define a local state value-function for agent $i$:
\begin{equation}
	\hat{V}^\pi_i(s')=\mathbb{E}\left[\hat{R}_i(s,a)|s_0=s'\right],
\end{equation}
where $\hat{R}_i(s,a)=\sum_{t=0}^{T-1} \gamma^t\hat{r}_{i,t}$.

Now we are able to obtain a local value-function for each agent $i$ to maximize:
\begin{equation}
	\hat{J}_i(\theta)=\mathbb{E}_{s_0\sim\mathcal{D}}\hat{V}^{\pi(\theta)}_i(s_{0})=\sum_{j\in\mathcal{I}^L_i}J_j(\theta).
\end{equation}

Given a function $f(\theta):\mathbb{R}^d\rightarrow\mathbb{R}$ and a positive $\delta$, we define
\begin{equation}
	f^\delta(\theta)=\mathbb{E}[f(\theta+\delta u)], ~~u\sim\mathcal{N}(0,I_d).
\end{equation}
 The following lemma builds a connection between the local value-function $\hat{J}_i^\delta(\theta)$ and the global value-function $J^\delta(\theta)$.

\begin{lemma}\label{le lg=gg}
The following statements are true: 
	
(i)	$\nabla_{\theta_i}J^\delta(\theta)=\nabla_{\theta_i}\hat{J}_i^\delta(\theta)$ for any $r>0$, $i\in\mathcal{V}$.

(ii) If $J_i(\theta)$, $i\in\mathcal{V}$ are differentiable, then $\nabla_{\theta_i}J(\theta)=\nabla_{\theta_i}\hat{J}_i(\theta)$, $i\in\mathcal{V}$.
\end{lemma}
\begin{proof}	
		Define 
	\begin{equation}
		\bar{J}_i(\theta)=J(\theta)-\hat{J}_i(\theta)=\sum_{j\in\mathcal{V}\setminus\mathcal{I}_i^L}\mathbb{E}_{s_0\sim\mathcal{D}}V_j^{\pi(\theta)}(s_0). 
	\end{equation}
Next we show that $\bar{J}_i(\theta)$ is independent of $\theta_i$.
	
	Let $\mathcal{N}_i^{L-}$ be the set of vertices in graph $\mathcal{G}_C$ that can reach $i$. Note that for each agent $j$, its action at time $t$, i.e., $a_{j,t}$, is only affected by the partial observation $o_{j,t}$, the current state $s_{j,t}$, and policy $\theta_j$. Therefore, there exists a function $f_j:\mathcal{O}_j\times\mathbb{R}^{d_j}\rightarrow P(\mathcal{A}_j)$ such that 
	\begin{equation}\label{ajt}
		\begin{split}
			a_{j,t}\sim f_j(o_{j,t},\theta_j)=f_j(\{s_{k,t}\}_{k\in\mathcal{I}_j}, \theta_j).
		\end{split}		
	\end{equation}
	Similarly, according to the definition of $\mathcal{P}_i$, there exists another function $h_j:\Pi_{k\in\mathcal{I}_j}\mathcal{S}_k\times\Pi_{k\in\mathcal{I}_j}\mathcal{A}_k\rightarrow P(\mathcal{S}_j)$ such that
	\begin{equation}\label{sjt}
		s_{j,t}\sim h_j( \{s_{k,t-1}\}_{k\in\mathcal{I}_j}, \{a_{k,t-1}\}_{k\in\mathcal{I}_j}).
	\end{equation}
	
	According to (\ref{ajt}) and (\ref{sjt}), we conclude that $s_{j,t}$ and $a_{j,t}$ are essentially only affected by $\{s_{k,0}\}_{k\in\mathcal{N}_j^{L-}\cup\{j\}}$, $\{\theta_{k}\}_{k\in\mathcal{N}_j^{L-}\cup\{j\}}$, and the transition probability function of each agent $k\in\mathcal{N}_j^{L-}$. 
	
	Note also that if $j\notin\mathcal{I}_i^L$, then $i\notin\mathcal{N}^{L-}_j\cup\{j\}$, which implies that  $\theta_i$ never influences $\mathbb{E}[R_j(s,a)|s_0=s]$ for $j\notin\mathcal{N}_i^L$. Therefore, $\bar{J}_i(\theta)$ is independent of $\theta_i$.
	
Proof for (i):	it has been shown in \cite{nesterov2017random} that
\begin{equation}
	\nabla_\theta J^\delta(\theta)=\frac{1}{\delta\kappa}\int_{\mathbb{R}^d} J(\theta+\delta u)e^{-\frac12\|u\|^2}udu,
\end{equation}	
where 
\begin{equation}\label{kappa}
	\kappa=\int_{\mathbb{R}^d}e^{-\frac12\|u\|^2}du.
\end{equation}
Define \begin{equation}\label{Phii}
	\Phi^i=(\mathbf{0}_{d_i\times d_1},..., I_{d_i},...,\mathbf{0}_{d_i\times d_N})\in\mathbb{R}^{d_i\times d},
\end{equation}
then $u_i=\Phi^iu$. It follows that
\begin{equation}
	\begin{split}
		\nabla_{\theta_i}J^\delta(\theta)&=\Phi^i\nabla_\theta J^\delta(\theta)\\
		&=\frac{1}{\delta\kappa}\int_{\mathbb{R}^d} J(\theta+\delta u)e^{-\frac12\|u\|^2}u_idu\\
		&=\frac{1}{\delta\kappa}\int_{\mathbb{R}^d} \hat{J}_i(\theta+\delta u)e^{-\frac12\|u\|^2}u_idu\\
		&~~+\frac{1}{\delta\kappa}\int_{\mathbb{R}^d} \bar{J}_i(\theta+\delta u)e^{-\frac12\|u\|^2}u_idu.
	\end{split}
\end{equation}	
Let $\bar{\theta}_i=(...,\theta_j^\top,...)^\top_{j\neq i}\in\mathbb{R}^{d-d_i}$, $\bar{u}_i=(...,u_j^\top,...)^\top_{j\neq i}\in\mathbb{R}^{d-d_i}$. Since we have proved that $\bar{J}_i(\theta+\delta u)$ is independent of $u_i$, the following holds:
\begin{equation*}
	\begin{split}
		&\int_{\mathbb{R}^d} \bar{J}_i(\theta+\delta u)e^{-\frac12\|u\|^2}u_idu\\&
		=\int_{\mathbb{R}^{d-d_i}} \bar{J}_i(\bar{\theta}_i+\delta\bar{u}_i)e^{-\frac12\|\bar{u}_i\|^2}d\bar{u}_i\int_{\mathbb{R}^{d_i}} e^{-\frac12\|u_i\|^2}u_idu_i\\
		&=0.
	\end{split}	
\end{equation*}
Therefore, 
\begin{equation}
	\nabla_{\theta_i}J^\delta(\theta)=\frac{1}{\delta\kappa}\int_{\mathbb{R}^d} \hat{J}_i(\theta+\delta u)e^{-\frac12\|u\|^2}u_idu=\nabla_{\theta_i}\hat{J}_i^\delta(\theta).
\end{equation}

Proof for (ii): differentiability of  $J_i(\theta)$ for all $i\in\mathcal{V}$ implies that $\hat{J}_i(\theta)$ for all $i\in\mathcal{V}$ and $J(\theta)$ are differentiable as well. Since $\bar{J}_i(\theta)$ is independent of $\theta_i$, we have
\begin{equation}
	\nabla_{\theta_i}J(\theta)=\nabla_{\theta_i}(\hat{J}_i(\theta)+\bar{J}_i(\theta))=\nabla_{\theta_i}\hat{J}_i(\theta).
\end{equation}
	This completes the proof.
\end{proof}

It is important to note that statement (i) in Lemma \ref{le lg=gg} does not require $J_i(\theta)$, $i\in\mathcal{V}$ to be differentiable because $J^\delta(\theta)=\mathbb{E}[J(\theta+\delta u)]$ is always differentiable. When $J(\theta)$ is not differentiable, we choose to find the stationary point of $J^\delta(\theta)$. The gap between $J(\theta)$ and $J^\delta(\theta)$ can be bounded if $J(\theta)$ is Lipschitz continuous and $\delta>0$ is sufficiently small.

We make the following assumption for functions $J_i(\theta)$:
\begin{assumption}\label{as Lip}
	$J_i(\theta)$, $i\in\mathcal{V}$ are $L_i$-Lipschitz continuous in $\mathbb{R}^d$. That is, $|J_i(\theta)-J_i(\theta')|\leq L_i\|\theta-\theta'\|$ for any $\theta,\theta'\in\mathbb{R}^d$.
\end{assumption}

Assumption \ref{as Lip} implies that $J(\theta)$ is $L$-Lipschitz continuous in $\mathbb{R}^d$, where $L\triangleq\sum_{i\in\mathcal{V}}L_i$, due to the following fact:
\begin{equation}
	|J(\theta)-J(\theta')|\leq\sum_{i\in\mathcal{V}}|J_i(\theta)-J_i(\theta')|\leq \sum_{i\in\mathcal{V}}L_i\|\theta-\theta'\|.
\end{equation}

%\begin{assumption}\label{as local sm}
%	$\nabla_{\theta}J(\theta)$ and	$\nabla_\theta J_i(\theta)$, $i\in\mathcal{V}$  are $(\phi_{\theta},\beta_{\theta})$ locally Lipschitz in $\mathbb{R}^d$. 
%\end{assumption}

\subsection{Learning Graph Design}

Lemma \ref{le lg=gg} has shown that having the local gradient of a local value-function is sufficient for each agent to update its policy according to gradient ascent. In this paper, we will estimate the local gradient based on the function value.  To ensure that each agent $i$ has access to the value of $\hat{J}_i(\theta)$, we derive a learning graph (the communication graph required for learning) $\mathcal{G}_L=(\mathcal{V},\mathcal{E}_L)$, where
\begin{equation}\label{EL}
	\mathcal{E}_L=\{(j,i)\in\mathcal{V}\times\mathcal{V}: i\stackrel{\mathcal{E}_C}{\longrightarrow} j\}.
\end{equation}
Based on this definition, we have the following observations:
\begin{itemize}
	\item The subgraph corresponding to each cluster $\mathcal{V}_i$ in $\mathcal{G}_L$ is a clique.
	
	\item Once there is one communication link from one cluster to another cluster, then any pair of agents in these two clusters have a communication link.
	
	\item  The agents in the same cluster share the same local value-function.
\end{itemize}

To demonstrate the edge set definition (\ref{EL}), the learning graph corresponding to the coordination graph in Fig. \ref{fig coordination graph} is shown in Fig. \ref{fig learning graph}. According to the above-mentioned observations, the graph in Fig. \ref{fig learning graph} can be interpreted from a cluster perspective, as shown in Fig. \ref{fig cluster learning graph}. In real applications, one coordinator can be chosen in each cluster, based on which different clusters are able to exchange information with each other, so that a large amount of communication links in Fig. \ref{fig learning graph} can be avoided. 

The following assumption ensures that agent $i$ is able to communicate with those agents whose rewards influence $\hat{J}_i(\theta)$.

\begin{figure}
	\centering
	\includegraphics[width=8cm]{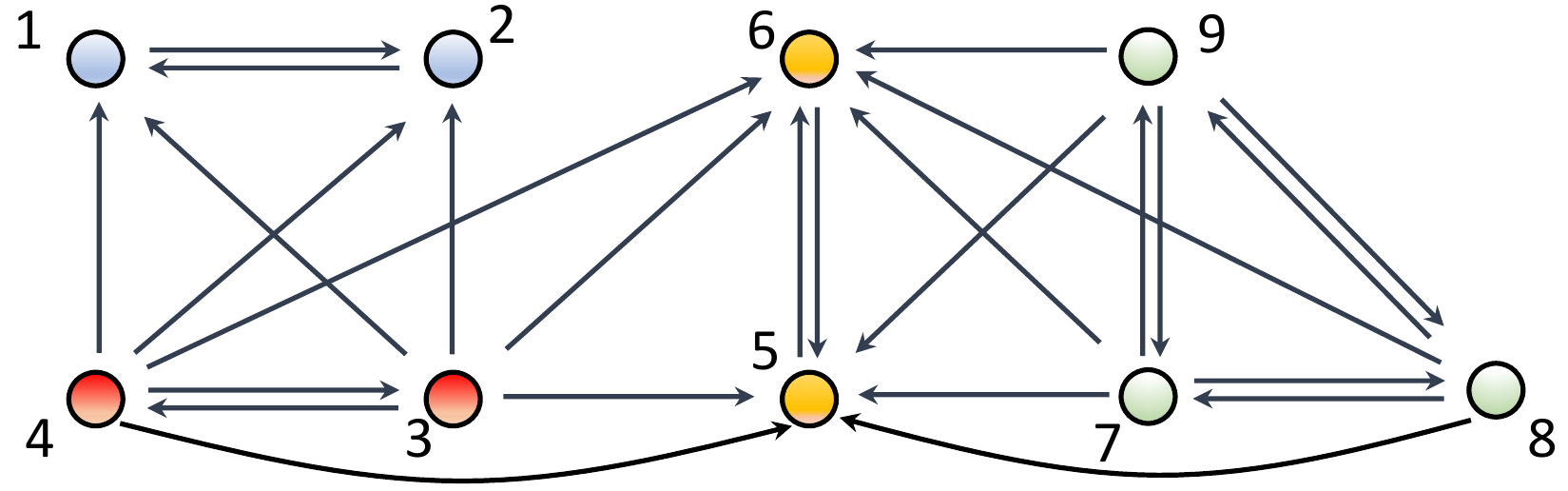}
	\caption{The learning graph corresponding to the coordination graph in Fig. \ref{fig coordination graph}.} \label{fig learning graph}
\end{figure}

\begin{figure}
	\centering
	\includegraphics[width=8cm]{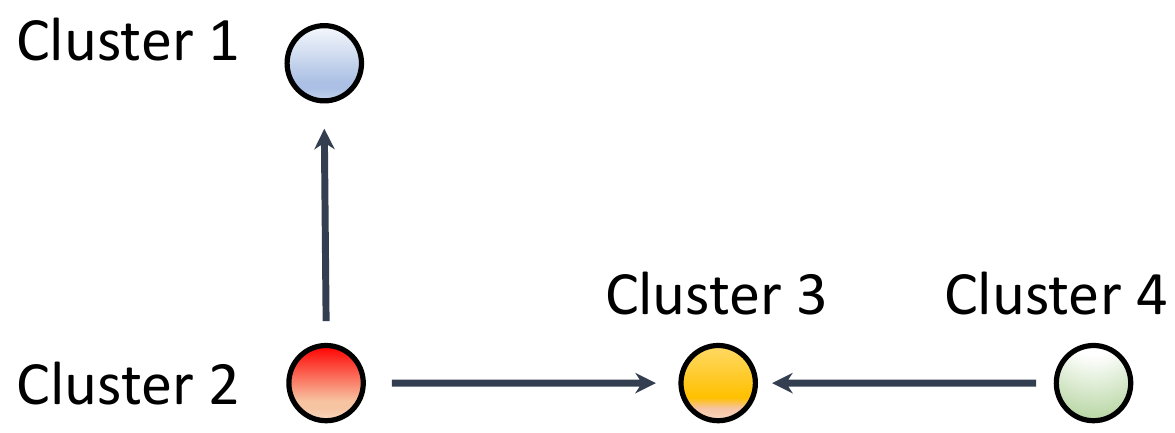}
	\caption{The cluster-wise learning graph.} \label{fig cluster learning graph}
\end{figure}

\begin{assumption}\label{as GL}
	During the learning process, agents are able to communicate with each other via the communication graph $\mathcal{G}_L$.
\end{assumption}

We assume that given the global policy $\pi(\theta)$ and initial state $s_0$, each agent is able to compute $R_i(s,a)$ and $\hat{R}_i(s,a)$ based on the received rewards $r_{j,t}$, $j\in\mathcal{I}_i^L$ at each time step $t=0,...,T-1$. By further considering random effects  in the transition, we rewrite the obtained values of agent $i$ by implementing policy $\pi(\theta)$ as $W_i(\theta,\xi)=
J_i(\theta)+\xi_i$ and $\hat{W}_i(\theta,\xi)=\hat{J}_i(\theta)+\hat{\xi}_i$ with $\hat{\xi}_i=\sum_{j\in\mathcal{I}_i^L}\xi_j$, where $\xi\in\mathbb{R}^N$ and $\xi\sim\mathcal{H}$, $\mathbb{E}[\xi_i]=0$, $\mathbb{E}[\xi_i^2]=\sigma_i^2$, $i\in\mathcal{V}$, $\mathcal{H}$ is determined by randomness of both $s_0$ and the transition probability. Therefore,
\begin{equation}\label{J=EW}
   J_i(\theta)=\mathbb{E}[W_i(\theta,\xi)], \hat{J}_i(\theta)=\mathbb{E}[\hat{W}_i(\theta,\xi)],~~~ \forall i\in\mathcal{V}.
\end{equation}

%is associated with $\mathcal{D}$. The following assumption is made on $\xi$:
%\begin{assumption}\label{as J=EW}
%	$\hat{J}_i(\theta)=\mathbb{E}[W_i(\theta,\xi)]$ for any $i\in\mathcal{V}$.
%\end{assumption}

Note that different from \cite{kar2013cal,zhang2018fully,zhang2020cooperative}, where the agents need to communicate with each other frequently during each learning episode, in our work, only one-time communication for transmitting $W_i(\theta^k,\xi^k)$ among neighbors in $\mathcal{G}_L$ is required in each learning episode $k$.

	\section{Distributed RL via ZOO} \label{sec algorithm}
	
In this section, we propose a distributed RL algorithm based on ZOO with policy search in the parameter space. Since the parameter space usually has a higher dimension than the action space, ZOO with policy search in the action space was proposed in \cite{kumar2020zeroth}. However, their approach requires the action space to be continuous and leverages the Jacobian of the policy $\pi$ w.r.t. $\theta$. Our distributed RL algorithm is applicable to both continuous and discrete action spaces and even does not require $\pi$ to be differentiable. In addition, our distributed learning framework based on local value-functions is actually also compatible with the method in \cite{kumar2020zeroth}.

\subsection{Distributed RL Based on Local Value Evaluation}
Let $\theta^k=(\theta^{k\top}_1,...,\theta^{k\top}_N)^\top\in\mathbb{R}^d$ be the global policy of the MAS at episode $k$. For agent $i$, its policy is updated at the $k$-th episode according to gradient ascent:
\begin{equation}\label{update}
	\theta_i^{k+1}=\theta_i^k+\eta g_i(\theta^k, u^k,\xi^k),
\end{equation}
where we adopt the following one-point zeroth-order oracle
\begin{equation}\label{gi}
	g_i(\theta^k, u^k,\xi^k)=\frac{\hat{W}_i(\theta^k+\delta u^k,\xi^k)}{\delta}u^k_i,
\end{equation}
$\xi^k\in\mathbb{R}^p$ is the random variable induced by random initial state $s_0^k$ and possible noises in agents' transitions at episode $k$, $\delta>0$ is the smoothing radius, $u^k\sim\mathcal{N}(0,I_d)$ is the randomly sampled perturbation direction, $u^k=(u^{k\top}_1,...,u^{k\top}_N)^\top$, and $u^k_i\in\mathbb{R}^{d_i}$.

The following lemma shows that $g_i(\theta^k, u^k,\xi^k)$ is an unbiased estimation of $\nabla_{\theta_i}J^\delta(\theta^k)$.
\begin{lemma}
Given policy $\theta^k$, it holds that	$\nabla_{\theta_i}J^\delta(\theta^k)=\mathbb{E}[g_i(\theta^k, u^k,\xi^k)]$.
\end{lemma}	
\begin{proof}
	According to (\ref{J=EW}), 
	\begin{equation}
		\begin{split}
			\mathbb{E}[g_i(\theta^k, u^k,\xi^k)]&=\mathbb{E}[\frac{\hat{J}_i(\theta^k+\delta u^k)}{\delta}u_i^k]\\
			&=\Phi^i \mathbb{E}[\frac{\hat{J}_i(\theta^k+\delta u^k)}{\delta}u^k]\\
			&=\Phi^i \nabla_{\theta}\hat{J}_i^\delta(\theta^k)\\
			&=\nabla_{\theta_i}\hat{J}_i^\delta(\theta^k),
		\end{split}
	\end{equation}
	where $\Phi^i$ is defined in (\ref{Phii}) and the third equality has been proved in \cite{nesterov2017random}. Recall Lemma \ref{le lg=gg}, the proof is completed.
\end{proof}	

\begin{remark}
	Besides	the one-point feedback gradient estimation (\ref{gi}), there are other efficient unbiased estimation approaches in the literature such as the two-point feedback oracle \cite{nesterov2017random} and residual feedback oracle \cite{zhang2020cooperative}, via which the second moment of the estimated gradient has an upper bound independent of the value of $W_i$. The distributed RL framework in this work is also compatible with those different gradient estimates, and the variance may be reduced by utilizing different zeroth-order orcacles. However, in real applications, it may be costly to implement one more policy to compute the long term accumulated reward $W_i$, or to store  information at each time step. Therefore, we adopt the simplest conventional one-point zeroth-order oracle (\ref{gi}).
\end{remark}

Based on the zeroth-order oracle (\ref{gi}), the distributed RL algorithm is proposed in Algorithm \ref{alg}.	
	\begin{algorithm}[htbp]
		\small
		\caption{Distributed RL Algorithm}\label{alg}
		\textbf{Input}: Step-size $\eta$, initial state distribution $\mathcal{D}$, number of learning epochs $K$, number of iterations $T$ (for policy evaluation), initial policy parameter $\theta_0$, smoothing radius $\delta>0$.\\
		\textbf{Output}: $\theta^K$.
		\begin{itemize}
			\item[1.] \textbf{for} $k=0,1,...,K-1$ \textbf{do}
			\item[2.] ~~~~Sample $s_0^k\sim\mathcal{D}$.
			\item[3.] ~~~~\textbf{for} all $i\in\mathcal{V}$ \textbf{do} (Simultaneous Implementation)
			\item[4.] ~~~~~Agent $i$ samples $u_{i}^k\sim\mathcal{N}(0,I_{d_i})$, implements policy $\pi_i(\theta_i^k+\delta u_i^k)$ for $t=0,...,T-1$, observes  $W_i(\theta^k+\delta u^{k}, \xi^k)$, computes $g_i(\theta^k, u^k,\xi^k)$ in (\ref{gi}), and then obtains $\theta_i^{k+1}$ according to (\ref{update}).
			\item[5.] ~~~~\textbf{end}
			\item[6.] \textbf{end}
		\end{itemize}
	\end{algorithm}	

\begin{remark}
	Note that when the coordination graph $\mathcal{G}_C$ is strongly connected, we have $\mathcal{I}_i^L=\mathcal{V}$ for all $i\in\mathcal{V}$, then Algorithm \ref{alg} becomes a centralized RL algorithm. In practice, when the multi-agent network is of a large scale, it is reasonable to consider that the whole coordination graph is weakly connected due to limited energy, selfishness, or physical constraints of partial agents. In such cases Algorithm \ref{alg} is more efficient than centralized approaches in terms of gradient variance and convergence rate.
\end{remark}

	\subsection{Convergence Result} \label{subsec convergence}

We further make the following assumption to ensure solvability of (\ref{maxJ}).	
\begin{assumption}\label{as J*}
	There exists $J_i^*>0$ such that $|J_i(\theta)|\leq J_i^*$ for any $\theta\in\mathbb{R}^d$, $\xi\in\mathbb{R}^p$, $i\in\mathcal{V}$.
\end{assumption}		
Assumption \ref{as J*} implicitly implies that there exists an optimal policy for the RL problem (\ref{maxJ}), which is the premise of solving problem (\ref{maxJ}). Based on Assumption \ref{as J*}, we can bound $|\hat{J}_i|$ and $|J|=|\sum_{i\in\mathcal{V}}J_i|$ as $\sum_{j\in\mathcal{I}_i^L}J_j^*\triangleq \hat{J}_i^*$ and $\sum_{i\in\mathcal{V}}J_i^*\triangleq J^*$, respectively. Note that boundedness of $J_i$ does not lead to boundedness of $W_i$. Therefore, the observed $W_i$ may be unbounded due to unbounded noises such as Gaussian noise during transition.

The following lemma bounds the variance of the zeroth-order oracle (\ref{gi}).
	
\begin{lemma}\label{le variance}
	Under Assumption \ref{as J*}, for any $\theta^k\in\mathbb{R}^d$, $\mathbb{E}[\|g_i(\theta^k, u^k,\xi^k)\|^2]\leq \frac{(\hat{J}_i^{*2}+\hat{\sigma}_i^2)d_i}{\delta^2}$, $i\in\mathcal{V}$, here $\hat{\sigma}_i^2=\sum_{j\in\mathcal{I}_i^L}\sigma_j^2$.
\end{lemma}	
\begin{proof} According to (\ref{gi}), we have
	\begin{align}
			&\mathbb{E}[\|g_i(\theta^k, u^k,\xi^k)\|^2]=\frac{\mathbb{E}[\hat{W}_i^2(\theta^k+\delta u^k,\xi^k)\|u^k_i\|^2]}{\delta^2}\\
			&=\frac{1}{\delta^2\kappa}\int_{\mathbb{R}^d}\mathbb{E}\left[(\hat{J}_i(\theta^k+\delta u^k)+\hat{\xi}_i^k)^2\right]\|u^k_i\|^2e^{-\frac12\|u^k\|^2}du^k\\
			&\leq \frac{1}{\delta^2\kappa}\int_{\mathbb{R}^d} (\hat{J}_i^{*2}+\mathbb{E}[\hat{\xi}_i^{k2}]) \|u^k_i\|^2e^{-\frac12\|u^k\|^2}du^k\\
			&\leq \frac{\hat{J}_i^{*2}+\hat{\sigma}_i^2}{\delta^2\kappa}\int_{\mathbb{R}^{d_i}}\|u_i^k\|^2e^{-\frac12\|u_i^k\|^2}du_i^k\int_{\mathbb{R}^{d-d_i}}e^{-\frac12\|v\|^2}dv\\
			&\leq\frac{\hat{J}_i^{*2}+\hat{\sigma}_i^2}{\delta^2\kappa}d_i\int_{\mathbb{R}^{d_i}}e^{-\frac12\|u_i^k\|^2}du_i^k\int_{\mathbb{R}^{d-d_i}}e^{-\frac12\|v\|^2}dv\\
			&=\frac{(\hat{J}_i^{*2}+\hat{\sigma}_i^2)d_i}{\delta^2},
	\end{align}
where $\kappa$ is  in (\ref{kappa}), the first inequality used the independence between different $u_i$, the second inequality used $\int_{\mathbb{R}^{d_i}}\|u_i^k\|^2e^{-\frac12\|u_i^k\|^2}du_i^k\leq d_i\int_{\mathbb{R}^{d_i}}e^{-\frac12\|u_i^k\|^2}du_i^k$, which has been proved in \cite[Lemma 1]{nesterov2017random}.
\end{proof}	
	
Lemma \ref{le variance} shows that the variance of each local zeroth-order 	oracle is only associated with the bound of the local value-function $\hat{J}^*_i$, and the variances of observations noises for agents in $\mathcal{I}_i^L$. Note that if the policy evaluation is based on the global reward, the bound of $\mathbb{E}[\|g_i(\theta^k, u^k,\xi^k)\|^2]$ would be $\frac{(J^{*2}+\sum_{i=1}^N\sigma_i^2)d_i}{\delta^2}$. When the network is of a large scale, $J^*$ may be much larger than $\hat{J}_i^*$, which means that our algorithm has a significantly improved scalability to large-scale networks.
	
\begin{theorem}\label{th Lp}
	Under Assumptions \ref{as Lip}-\ref{as J*}, let $\delta=\frac{\epsilon}{L\sqrt{d}}$, $\eta=\frac{\epsilon^{1.5}}{d^{1.5}\sqrt{K}}$. Then 
	
	(i). $|J^\delta(\theta)-J(\theta)|\leq \epsilon$ for any $\theta\in\mathbb{R}^d$.
	
	(ii). By implementing Algorithm \ref{alg}, if $K\geq \frac{d^3B^2}{\epsilon^5}$, then
	\begin{equation}\label{rate}
	\frac{1}{K}\sum_{k=0}^{K-1}\mathbb{E}[\|\nabla_{\theta}J^\delta(\theta^k)\|^2]\leq\epsilon,
	\end{equation}
where $B=J^*-J^\delta(\theta^0)+L^4 (J_0^{2}+\sigma_0^2)/2$, $J_0=\max_{i\in\mathcal{V}}\hat{J}_i^*$, $\sigma_0=\max_{i\in\mathcal{V}}\hat{\sigma}_i$.
\end{theorem}	
\begin{proof}
 Statement (i) is obtained by using the Lipschitz continuity of $J(\theta)$, which has been proved in \cite{nesterov2017random}. 
 
 Now we prove statement (ii). According to Assumption \ref{as Lip} and \cite[Lemma 2]{nesterov2017random}, the gradient of $J^\delta(\theta)$ is $\sqrt{d}L/\delta$-Lipschitz continuous. Let $g^k=(g_1^\top(\theta^k,u^k,\xi^k),...,g_N^\top(\theta^k,u^k,\xi^k))^\top\in\mathbb{R}^d$, the following holds:
	\begin{multline}
		|J^\delta(\theta^{k+1})- J^\delta(\theta^k)-\langle\nabla_{\theta}J^\delta(\theta^k),\eta g(\theta^k)\rangle|
		\leq\frac{\sqrt{d}L}{2\delta}\eta^2\|g^k\|^2,
	\end{multline}
	which implies that
	\begin{equation}\label{nablaJetag}
	\langle\nabla_{\theta}J^\delta(\theta^k),\eta g(\theta^k)\rangle\leq J^\delta(\theta^{k+1})- J^\delta(\theta^k)+\frac{\sqrt{d}L}{2\delta}\eta^2\|g^k\|^2.
	\end{equation}

Lemma \ref{le variance} implies that 
\begin{equation}
	\begin{split}
		&\mathbb{E}[\|g(\theta^k, u^k,\xi^k)\|^2]=\sum_{i=1}^N \mathbb{E}[\|g_i(\theta^k, u^k,\xi^k)\|^2]\\ &\leq \sum_{i=1}^N(2J_0^{2}+2\sigma_0^2)d_i/\delta^2=(2J_0^{2}+2\sigma_0^2)d/\delta^2.
	\end{split}
\end{equation}
Taking expectation on both sides of (\ref{nablaJetag}), we obtain 
	\begin{align*}
			&\eta\mathbb{E}[\|\nabla_{\theta}J^\delta(\theta^k)\|^2]\\
			&\leq \mathbb{E}[J^\delta(\theta^{k+1})- J^\delta(\theta^k)]+\frac{\sqrt{d}L}{2\delta}\eta^2\mathbb{E}[\|g(\theta^k,u^k,\xi^k)\|^2]\\
			&\leq \mathbb{E}[J^\delta(\theta^{k+1})- J^\delta(\theta^k)] +\frac{L\eta^2(J_0^{2}+\sigma_0^2)d^{1.5}}{2\delta^3}.
	\end{align*}
It follows that
\begin{align}
&\frac{1}{K}\sum_{k=0}^{K-1}\mathbb{E}[\|\nabla_{\theta}J^\delta(\theta^k)\|^2]\\ &\leq\frac{1}{\eta}\left[\frac1K\left(\mathbb{E}[J^\delta(\theta^K)]-J^\delta(\theta^0)\right)+\frac{L\eta^2 (J_0^{2}+\sigma_0^2)d^{1.5}}{2\delta^3}\right]\\
&\leq\frac{d^{1.5}}{\epsilon^{1.5} \sqrt{K}}\left[ (J^*-J^\delta(\theta^0))+L^4 (J_0^{2}+\sigma_0^2)/2\right].\label{averagenabla}
\end{align}
The proof is completed.	
\end{proof}	

\textbf{Advantage Analysis.} If the policy evaluation is based on global value-function $J(\theta)$, then $J_0^2$ and $\sigma_0^2$ in (\ref{averagenabla}) will be replaced by $J^{*2}$ and $\sum_{i=1}^N\sigma_i^2$, which may dominate the complexity and be significantly larger than $J_0^2$ and $\sigma_0^2$ when there are a large number of agents. This implies that our algorithm has a higher scalability to large-scale network problems than the algorithms with policy evaluation based on the global value-function.

According to (\ref{averagenabla}), the convergence result can also be written as $\frac{1}{K}\sum_{k=0}^{K-1}\mathbb{E}[\|\nabla_{\theta}J^\delta(\theta^k)\|^2]\leq \mathcal{O}(Bd^{1.5}\epsilon^{-1.5}K^{-0.5})$, which is similar to the result in \cite[Theorem 3.2]{zhang2020cooperative}. However, they require a consensus algorithm to be run for a sufficiently long time during each learning episode, which is unnecessary in our algorithm. Therefore, our algorithm has a lower sample complexity. However, the residual feedback used in \cite{zhang2020cooperative} may exhibit a lower gradient variance than the one-point feedback (\ref{gi}) when $\hat{J}^*_i$ is very large for each agent $i$. In the next subsection, we show that if the two-point feedback or the residual feedback is utilized in our work, our local value-function design can further reduce the gradient estimation variance for large-scale networks.

\subsection{Discussions on Distributed RL with Two-Point Feedback} \label{subsec two-point}

We have proposed a distributed RL algorithm based on the one-point zeroth-order oracle (\ref{gi}). We observe that Algorithm \ref{alg} is always efficient as long as $g(\theta^k,u^k,\xi^k)$ is an unbiased estimation of $\nabla_\theta J^\delta(\theta^k)$ and $\mathbb{E}[\|g(\theta^k,u^k,\xi^k)\|^2]$ is bounded. Therefore, the two-point feedback oracles proposed in \cite{nesterov2017random} and the residual feedback oracle in \cite{zhang2020improving} can also be employed in Algorithm \ref{alg}. Based on the local value-function design in our work, the two-point feedback oracle and the residual feedback oracle can be obtained by
\begin{equation}\label{gbar}
	\bar{g}_i(\theta^k, u^k,\xi^k)=\frac{\hat{W}_i(\theta^k+\delta u^k,\xi^k)-\hat{W}_i(\theta^k,\xi^k)}{\delta}u^k_i,
\end{equation}
and
\begin{equation}
	\small
	\tilde{g}_i(\theta^k, u^k,\xi^k)=\frac{\hat{W}_i(\theta^k+\delta u^k,\xi^k)-\hat{W}_i(\theta^{k-1}+\delta u^{k-1},\xi^{k-1})}{\delta}u^k_i,
\end{equation}
respectively, where only local value-function is needed, in contrast to \cite{nesterov2017random,zhang2020improving}. 

Define $\hat{L}_i=\sum_{j\in\mathcal{I}_i^L}L_i$, which is actually a Lipschitz constant of $\hat{W}_i$. Then we are able to show
\begin{align}
		&\mathbb{E}[\|\bar{g}_i(\theta^k,u^k,\xi^k)\|^2]\leq \mathbb{E}[(\hat{L}_i^2+\hat{\sigma}_i^2) \|u^k\|^2 \|u_i^k\|^2]\\
		&\leq (\hat{L}_i^2+\hat{\sigma}_i^2)\left(\mathbb{E}[\|u_i\|^4]+\mathbb{E}[\sum_{j\in\mathcal{V}\setminus\{i\}}\|u_j\|^2]\mathbb{E}[\|u_i\|^2]\right)\\
		&\leq (\hat{L}_i^2+\hat{\sigma}_i^2)\left((d_i+4)^2+(d-d_i)d_i\right)\\
		&=(\hat{L}_i^2+\hat{\sigma}_i^2)(d_id+8d_i+16).
\end{align}
As a result,
\begin{equation}\label{boundgbar}
	\begin{split}
		\mathbb{E}[\|\bar{g}(\theta^k,u^k,\xi^k)\|^2]
		&\leq (L_0^2+\sigma_0^2)(d+4)^2,
	\end{split}
\end{equation}
where $L_0=\max_{i\in\mathcal{V}}\hat{L}_i$.

\textbf{Advantage Analysis.} If the two-point feedback oracle (\ref{gbar}) is based on  the global observation $W(\theta,\xi)=J(\theta)+\sum_{i=1}^N\xi_i$, which is the case in \cite{nesterov2017random}, the bound in (\ref{boundgbar}) will be increased by replacing $L_0^2$ and $\sigma_0^2$ with $L^2$ and $\sum_{i=1}^N\sigma^2$. Similarly, for the residual feedback, we can show that the bound of $\mathbb{E}[\|\tilde{g}_i(\theta^k,u^k,\xi^k)\|^2]$ can be reduced by using local value-functions due to the reduction of the Lipschitz constant from $L$ to $L_0$. In conclusion, a local value-function always exhibits a lower variance of the zeroth-order oracle than the global value-function, thereby inducing a lower sample complexity.

	%%%%%%%%%%%%%%%%%%%%%%%%%%%%%%%%%%%%%%%%%%%%%%%%%%%%%%%%%%%%%%%%%%%%%%%%%%%%%%%%
	
\section{A Distributed Resource Allocation Example}	\label{sec sim}
	
We adopt a generalized version of the example in \cite{zhang2020cooperative}, where the coordination graph is directed and the partial observation of each agent includes states of its collaborative neighbors and itself instead of only itself. Nonetheless, our algorithm is applicable to the case when the partial observation of each agent is itself only.
	
Consider a distributed resource allocation problem with a coordination graph $\mathcal{G}_C=(\mathcal{V},\mathcal{E}_C)$. There are $N$ agents representing $N$ warehouses. At time step $t$, each agent $i$ stores resources with the amount $m_i(t)\in\mathbb{R}$,  receives a local demand $d_i(t)\in\mathbb{R}$, sends partial of its resources to and receives resources from its neighbors $j\in\mathcal{N}_i$. Then agent $i$ has the following dynamics
	\begin{equation}
		\begin{split}
			m_i(t+1)&=m_i(t)-\sum_{j\in\mathcal{N}_i^{out}}a_{ij}(t)m_i(t)\\
			&~~~~+\sum_{j\in\mathcal{N}_i}a_{ji}(t)m_j(t)-d_i(t),\\
			d_i(t)&=A_i(1-\sin(w_{i,t}t))+w_{i,t},
		\end{split}
	\end{equation}
	where $a_{ij}(t)\in[0,1]$ denotes the fraction of resources agent $i$ sends to its neighbor $j$ at time $t$, $0<A_i<m_i(0)$ is a constant, $w_{i,t}$ is a bounded noise with zero mean, $\mathcal{N}_i=\{j\in\mathcal{V}: (j,i)\in\mathcal{E}_C\}$ and  $\mathcal{N}_i^{out}=\{j\in\mathcal{V}:(i,j)\in\mathcal{E}_C\}$ are the in-neighbor set and the out-neighbor set of agent $i$, respectively.
	
	Under Assumption \ref{as GL} and a learning graph $\mathcal{G}_L=(\mathcal{V},\mathcal{E}_L)$ with $\mathcal{E}_L$ defined in (\ref{EL}),  agent $i$ receives the following information at time step $t$ during the learning process:
	
	\textbf{Partial observation:} $o_{i,t}=[m_{\mathcal{I}_i}^\top(t),d_i^\top(t)]^\top\in\mathbb{R}^{|\mathcal{I}_i|+1}$.
	
	\textbf{Individual reward:} $r_{i,t}=0$ if $m_i(t)\geq0$, and $r_{i,t}=-m_i^2(t)$ otherwise.
	
	\textbf{Individual rewards of neighbors in $\mathcal{G}_L$:} $r_{j,t}$, $j\in\mathcal{N}_i^L$.
	
	Define $a_{i,t}=(...,a_{ij}(t),...)^\top_{j\in\mathcal{N}_i^{out}}\in\mathbb{R}^{|\mathcal{N}_i^{out}|}$ as the action agent $i$ takes. The local policy $\pi_i:\mathbb{R}^{|\mathcal{I}_i|+1}\rightarrow\mathbb{R}^{|\mathcal{N}_i|}$ maps the observation $o_{i,t}$ of agent $i$ to its action $a_{i,t}$. The global policy $\pi:\mathbb{R}^{18}\rightarrow\mathbb{R}^{\sum_{i\in\mathcal{V}}|\mathcal{N}_i|}$ maps all the agents' states to all the local policies.
	
	The cooperative control objective is to find an optimal policy $\pi^*$ such that the following global reward is maximized:
	\begin{equation}
		J=\sum_{i=1}^N\sum_{t=0}^{T-1}\gamma^t r_{i,t}.
	\end{equation}
	
	To seek the optimal policy $\pi_i(o_i)$ for agent $i$ to determine its action $\{a_{ij}\}_{j\in\mathcal{N}_i}$, we parameterize the policy function by defining
	\begin{equation}
		a_{ij}=\frac{\exp(-z_{ij})}{\sum_{j\in\mathcal{I}_i}\exp(-z_{ij})},
	\end{equation}
	where $z_{ij}$ is approximated by radial basis functions:
	\begin{equation}
		z_{ij}=\sum_{l=1}^{n_c}\|o_i-c_{il}\|^2\theta_{ij}(l),
	\end{equation}
	$c_{il}=(\hat{c}_{il}^\top,\bar{c}_{il})\in\mathbb{R}^{|\mathcal{I}_i|+1}$ is the center of the $l$-th feature for agent $i$, here $\hat{c}_{il}\in\mathbb{R}^{|\mathcal{I}_i|}$ and $\bar{c}_{il}$ are set according to the ranges of $m_{\mathcal{I}_i}$ and $d_i$, respectively, such that $c_{il}$, $l=1,...,n_c$ are approximately evenly distributed in the range of $o_i$.

\begin{example}\label{ex 9warehouses}
Consider 9 warehouses with the coordination graph shown in Fig. \ref{fig coordination graph}. Fig. \ref{fig learning graph} shows the corresponding learning graph.
Set $m_i(0)=1+w_i$ for all $i=1,...,9$, $w_i$ is a random variable with zero mean and is bounded by $0.01$, iteration number $T=8$, number of learning epochs $K=600$, $d_i=0.2(1-\sin (wt))+wt$, $w\sim\mathcal{N}(0,0.01)$. Fig. \ref{Fig.4} (left) depicts the evolution of the observed values of the global value-function by implementing 4 different algorithms repeatedly for 10 times. Here $\xi^*$ corresponds to a fixed initial state $m_i(0)=1$, $i=1,...,9$. In each time of implementation, one perturbation vector $u^k$ is sampled and used for all the 4 algorithms during each learning episode $k$. The centralized algorithm is the zeroth-order optimization algorithm based on global value evaluation, while the distributed algorithm is based on local value evaluation (Algorithm \ref{alg}). The distributed two-point feedback algorithm is Algorithm \ref{alg} with $g_i(\theta^k,u^k,\xi^k)$ replaced by $\bar{g}_i(\theta^k,u^k,\xi^k)$ in (\ref{gbar}). We observe that the distributed algorithms are always faster than the centralized algorithms. Fig. \ref{Fig.4} (middle) and Fig. \ref{Fig.4} (right) show the comparison of centralized and distributed one-point feedback algorithms, and the comparison of centralized and distributed two-point feedback algorithms, respectively. From these two figures, it is clear that the distributed algorithms always exhibit lower variances in contrast to the centralized algorithms. This implies that the policy evaluation based on local value-functions is more robust than policy evaluation based on the global value-function.

\begin{figure*}[htbp]
\centering
\subfloat{\includegraphics[width=0.33\textwidth]{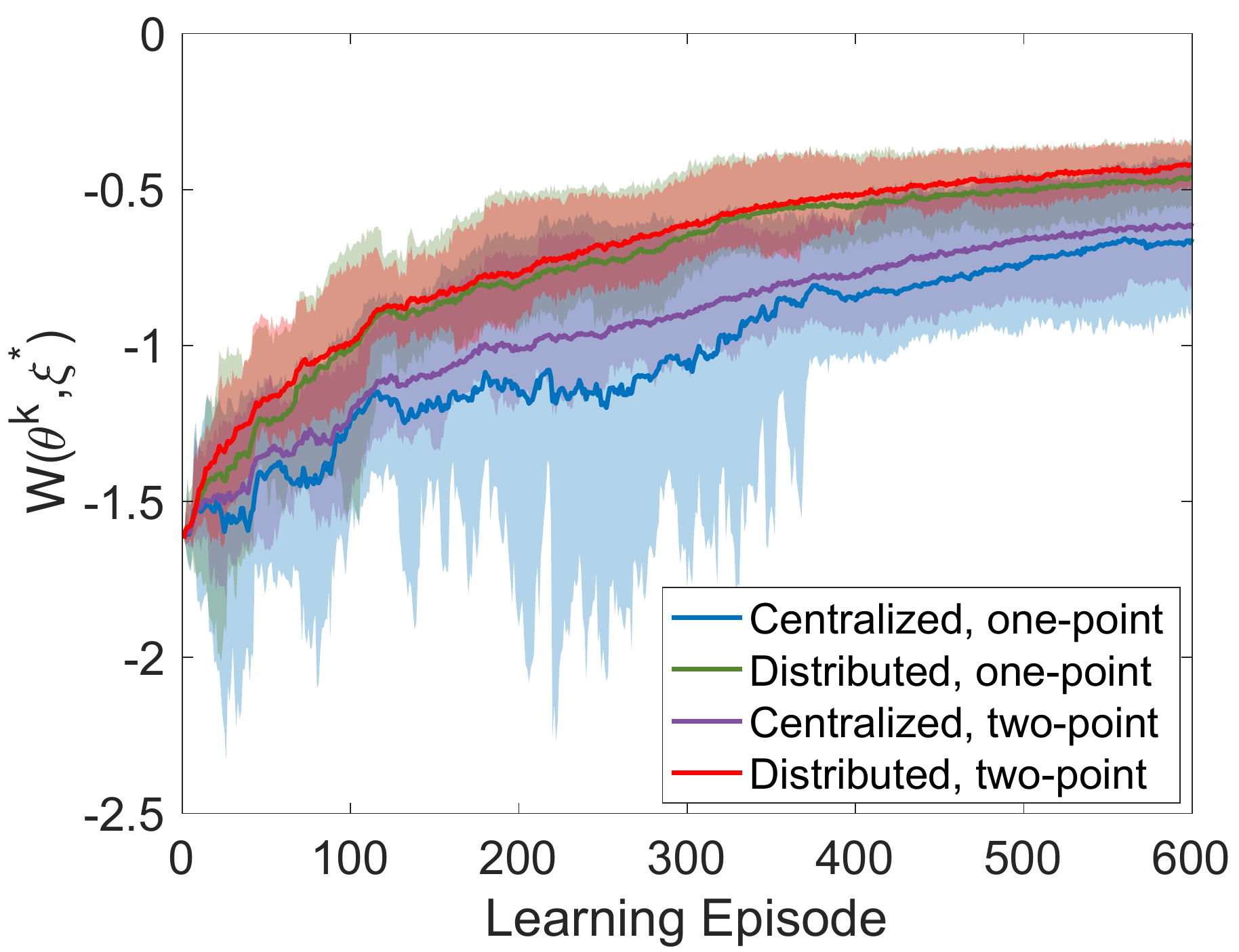}%
\label{}} \hfill
\subfloat{\includegraphics[width=0.33\textwidth]{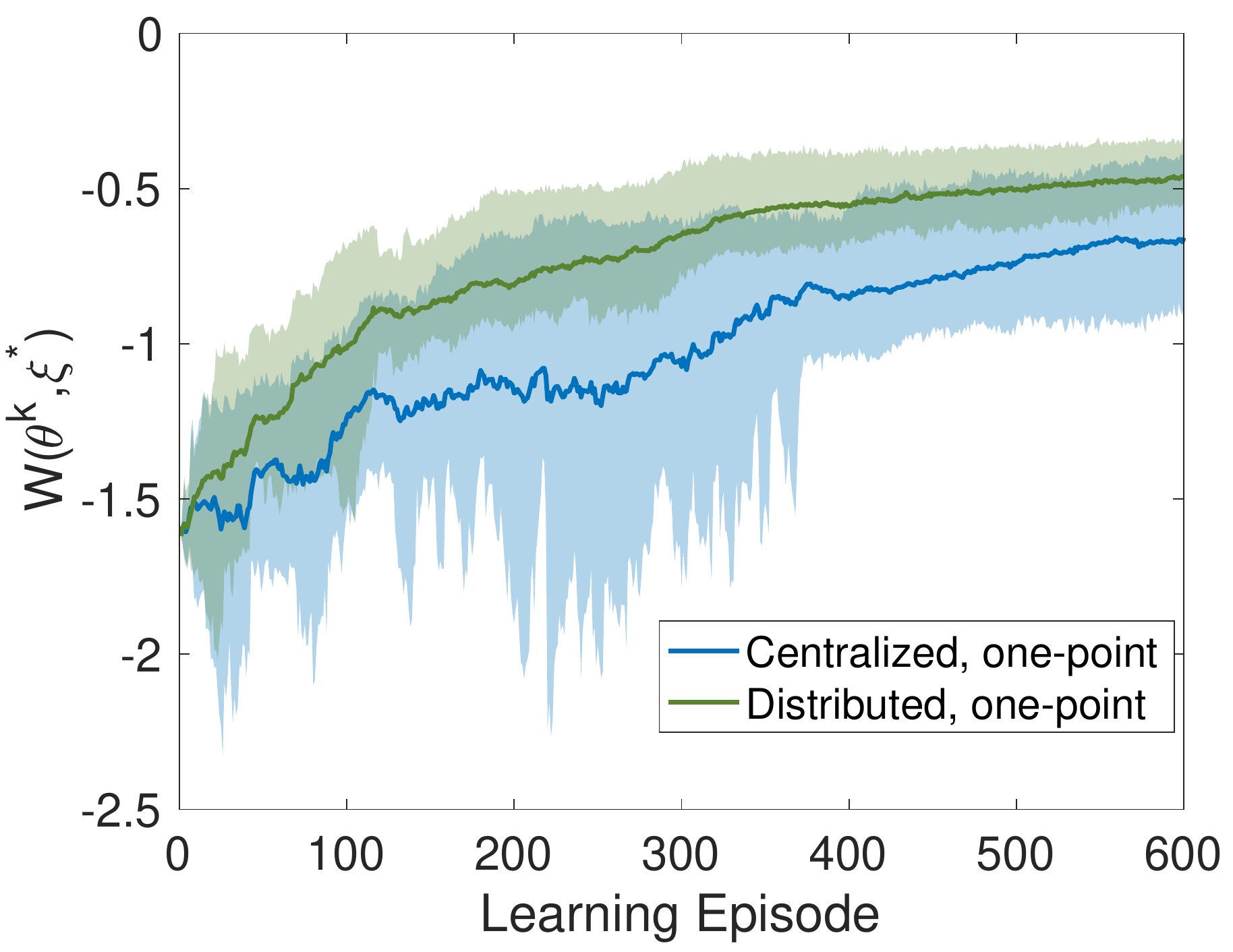}%
\label{}} \hfill
\subfloat{\includegraphics[width=0.33\textwidth]{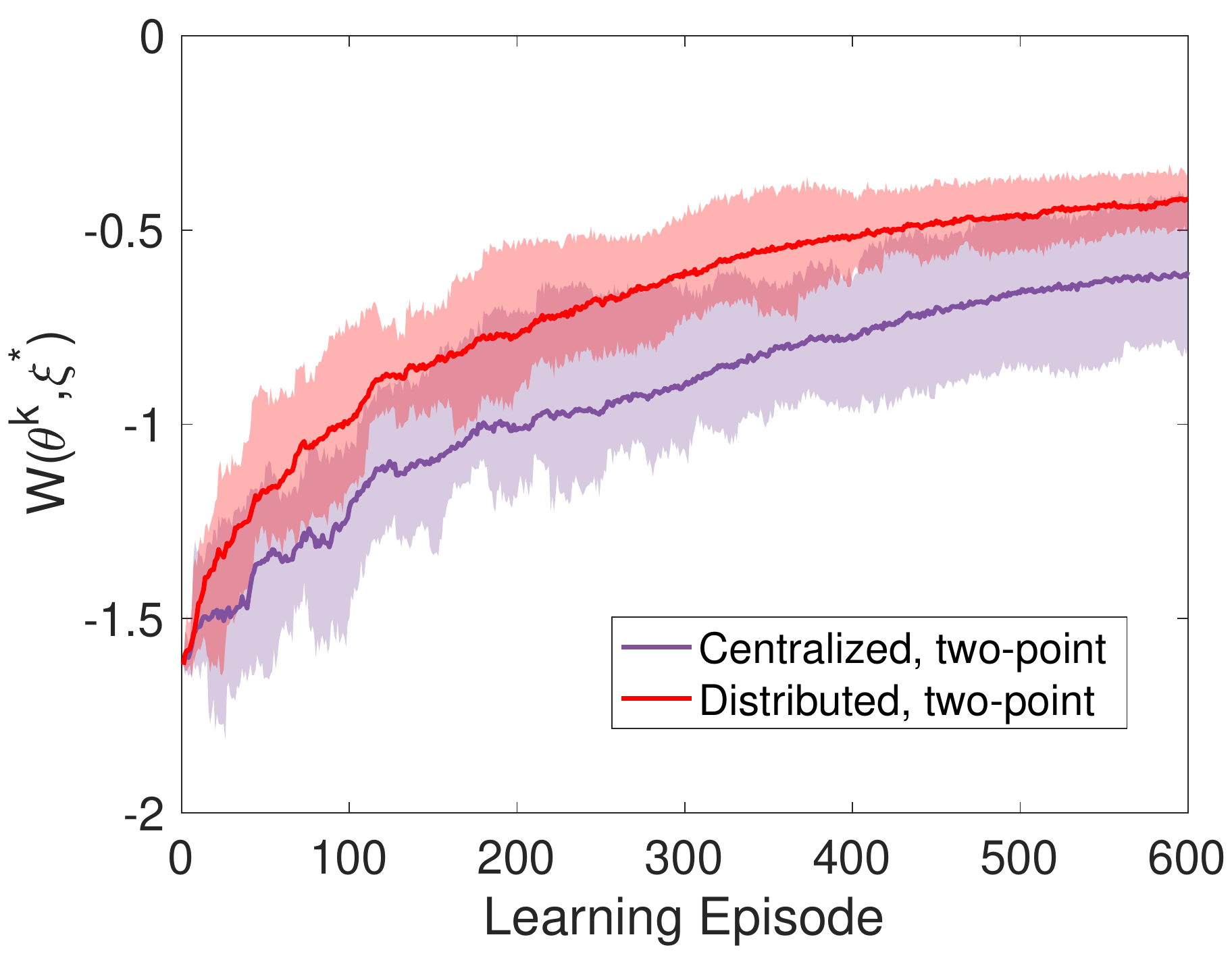}%
\label{}}	
\hspace*{0in}
\caption{Comparison of \textbf{(Left)} different algorithms for 9 warehouses;  \textbf{(Middle)} centralized and distributed algorithms under zeroth-order oracles with one-point feedback; \textbf{(Right)} centralized and distributed algorithms under zeroth-order oracles with two-point feedback.} 
\label{Fig.4}
\end{figure*} 

%\begin{figure}
%	\centering
%	\includegraphics[width=8cm]{4algorithms}
%	\caption{Comparison of different algorithms for 9 warehouses.} \label{fig simresult}
%\end{figure}

%\begin{figure}
%	\centering
%	\includegraphics[width=8cm]{onepoint}
%	\caption{Comparison of centralized and distributed algorithms under zeroth-order oracles with one-point feedback.} \label{fig onepoint}
%\end{figure}

%\begin{figure}
%	\centering
%	\includegraphics[width=8cm]{twopoint}
%	\caption{Comparison of centralized and distributed algorithms under zeroth-order oracles with two-point feedback.} \label{fig twopoint}
%\end{figure}

\end{example}

\begin{example}
We further consider a larger-size case where there are $N=100$ warehouses. The coordination graph $\mathcal{G}_C=(\mathcal{V},\mathcal{E}_C)$ has the following edge set: 
\begin{equation*}
    \mathcal{E}_C=\{(i,j)\in\mathcal{V}^2: |i-j|=1, i=2k-1, k\in\mathbb{N}\}\cup\{(1,N)\}.
\end{equation*}
Accordingly, graph $\mathcal{G}_L=(\mathcal{V},\mathcal{E}_L)$ has the edge set 
\begin{equation*}
    \mathcal{E}_L=\{(i,j)\in\mathcal{V}^2: |i-j|=1, i=2k, k\in\mathbb{N}\}\cup\{(N,1)\}.
\end{equation*}
The parameter settings are the same as those in Example \ref{ex 9warehouses} Fig. \ref{fig 100agents} shows the evolutions of the observed values of the global value-function by implementing 4 different algorithms repeatedly for 10 times. Observe that when the network is of a large scale, policy evaluation based on the global value-function becomes very unstable due to the high variance of the gradient estimates. On the contrary, our distributed algorithms based on observing local value-functions are more robust and exhibit significantly lower variances. This is consistent with our analysis in Subsections \ref{subsec convergence} and \ref{subsec two-point}.

\begin{figure}
	\centering
	\includegraphics[width=6cm]{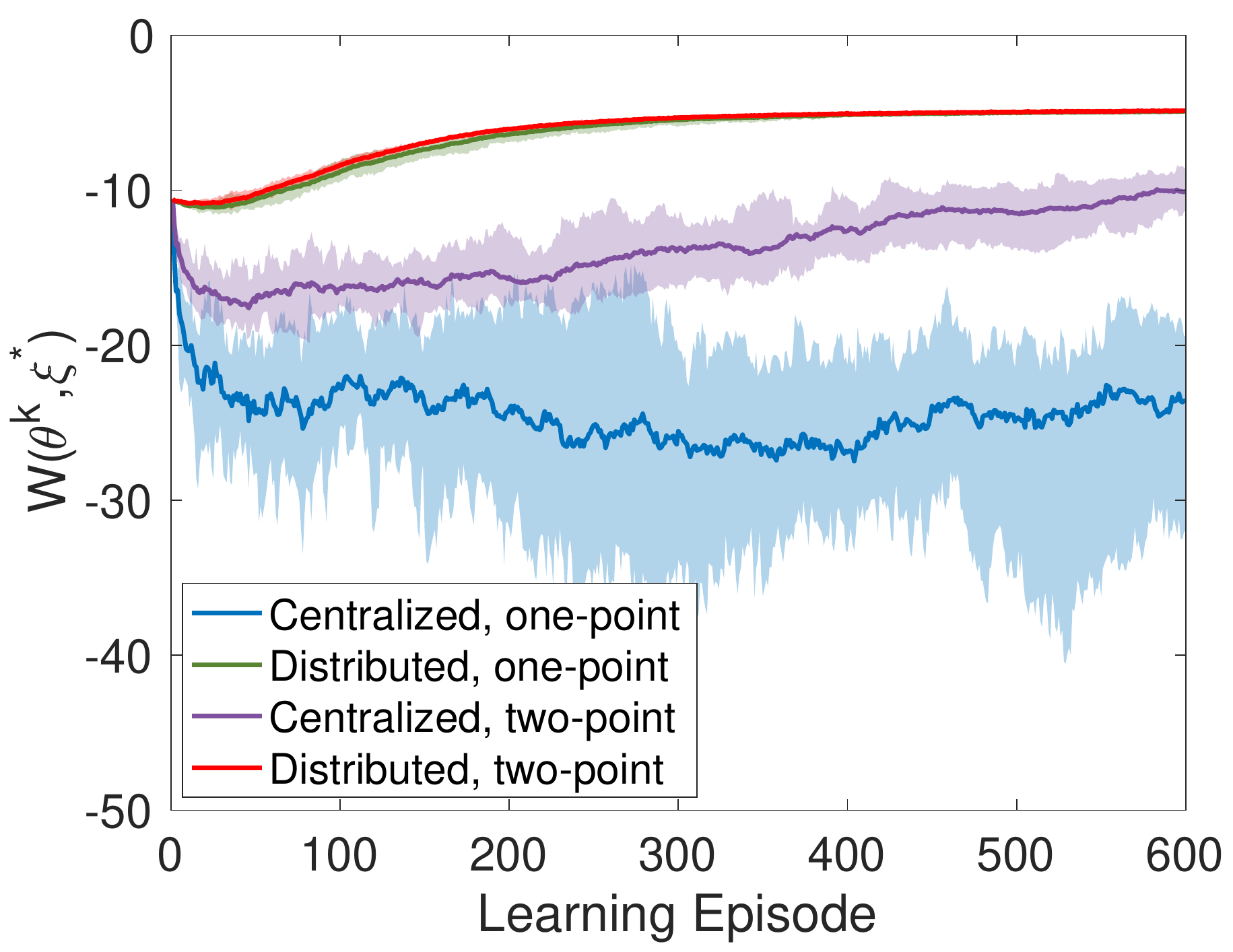}
	\caption{Comparison of different algorithms for 100 warehouses.} \label{fig 100agents}
\end{figure}

\end{example}

	\section{Conclusions}\label{sec: conclusion}
	
We studied the MARL problem with a directed coordination graph via a MDP formulation. Based on the weakly connected coordination graph, we showed that a local value-function can be designed for each agent, which is sufficient for estimating the local gradient of the global value-function. To ensure that each agent has access to the value of the local value-function, we further specified how the learning graph can be designed. By employing zeroth-order oracles, we proposed a distributed RL algorithm that has guaranteed convergence, and exhibits reduced variance of the gradient. The algorithm was applied to a distributed resource allocation problem. For our future work we would like to investigate how to reduce the number of required communication links using graph-theoretic reduction techniques.

\bibliography{Reference.bib}

\begin{thebibliography}{10}
\providecommand{\url}[1]{#1}
\csname url@samestyle\endcsname
\providecommand{\newblock}{\relax}
\providecommand{\bibinfo}[2]{#2}
\providecommand{\BIBentrySTDinterwordspacing}{\spaceskip=0pt\relax}
\providecommand{\BIBentryALTinterwordstretchfactor}{4}
\providecommand{\BIBentryALTinterwordspacing}{\spaceskip=\fontdimen2\font plus
\BIBentryALTinterwordstretchfactor\fontdimen3\font minus
  \fontdimen4\font\relax}
\providecommand{\BIBforeignlanguage}[2]{{%
\expandafter\ifx\csname l@#1\endcsname\relax
\typeout{** WARNING: IEEEtran.bst: No hyphenation pattern has been}%
\typeout{** loaded for the language `#1'. Using the pattern for}%
\typeout{** the default language instead.}%
\else
\language=\csname l@#1\endcsname
\fi
#2}}
\providecommand{\BIBdecl}{\relax}
\BIBdecl

\bibitem{kar2013cal}
S.~Kar, J.~M. Moura, and H.~V. Poor, ``${{\cal Q}{\cal D}} $-learning: A
  collaborative distributed strategy for multi-agent reinforcement learning
  through ${\rm consensus}+{\rm innovations} $,'' \emph{IEEE Transactions on
  Signal Processing}, vol.~61, no.~7, pp. 1848--1862, 2013.

\bibitem{zhang2018fully}
K.~Zhang, Z.~Yang, H.~Liu, T.~Zhang, and T.~Basar, ``Fully decentralized
  multi-agent reinforcement learning with networked agents,'' in
  \emph{International Conference on Machine Learning}.\hskip 1em plus 0.5em
  minus 0.4em\relax PMLR, 2018, pp. 5872--5881.

\bibitem{zhang2020cooperative}
Y.~Zhang and M.~M. Zavlanos, ``Cooperative multi-agent reinforcement learning
  with partial observations,'' \emph{arXiv preprint arXiv:2006.10822}, 2020.

\bibitem{chen2021communication}
T.~Chen, K.~Zhang, G.~B. Giannakis, and T.~Basar, ``Communication-efficient
  policy gradient methods for distributed reinforcement learning,'' \emph{IEEE
  Transactions on Control of Network Systems}, 2021.

\bibitem{cassano2020multiagent}
L.~Cassano, K.~Yuan, and A.~H. Sayed, ``Multiagent fully decentralized value
  function learning with linear convergence rates,'' \emph{IEEE Transactions on
  Automatic Control}, vol.~66, no.~4, pp. 1497--1512, 2020.

\bibitem{gronauer2021multi}
S.~Gronauer and K.~Diepold, ``Multi-agent deep reinforcement learning: a
  survey,'' \emph{Artificial Intelligence Review}, pp. 1--49, 2021.

\bibitem{qu2020scalable}
G.~Qu, A.~Wierman, and N.~Li, ``Scalable reinforcement learning of localized
  policies for multi-agent networked systems,'' in \emph{Learning for Dynamics
  and Control}.\hskip 1em plus 0.5em minus 0.4em\relax PMLR, 2020, pp.
  256--266.

\bibitem{littman1994markov}
M.~L. Littman, ``Markov games as a framework for multi-agent reinforcement
  learning,'' in \emph{Machine learning proceedings}.\hskip 1em plus 0.5em
  minus 0.4em\relax Elsevier, 1994, pp. 157--163.

\bibitem{guestrin2002coordinated}
C.~Guestrin, M.~Lagoudakis, and R.~Parr, ``Coordinated reinforcement
  learning,'' in \emph{ICML}, vol.~2.\hskip 1em plus 0.5em minus 0.4em\relax
  Citeseer, 2002, pp. 227--234.

\bibitem{kok2006collaborative}
J.~R. Kok and N.~Vlassis, ``Collaborative multiagent reinforcement learning by
  payoff propagation,'' \emph{Journal of Machine Learning Research}, vol.~7,
  pp. 1789--1828, 2006.

\bibitem{vemula2019contrasting}
A.~Vemula, W.~Sun, and J.~Bagnell, ``Contrasting exploration in parameter and
  action space: A zeroth-order optimization perspective,'' in \emph{The 22nd
  International Conference on Artificial Intelligence and Statistics}.\hskip
  1em plus 0.5em minus 0.4em\relax PMLR, 2019, pp. 2926--2935.

\bibitem{kumar2020zeroth}
H.~Kumar, D.~S. Kalogerias, G.~J. Pappas, and A.~Ribeiro, ``Zeroth-order
  deterministic policy gradient,'' \emph{arXiv preprint arXiv:2006.07314},
  2020.

\bibitem{nesterov2017random}
Y.~Nesterov and V.~Spokoiny, ``Random gradient-free minimization of convex
  functions,'' \emph{Foundations of Computational Mathematics}, vol.~17, no.~2,
  pp. 527--566, 2017.

\bibitem{zhang2020improving}
Y.~Zhang, Y.~Zhou, K.~Ji, and M.~M. Zavlanos, ``Improving the convergence rate
  of one-point zeroth-order optimization using residual feedback,'' \emph{arXiv
  preprint arXiv:2006.10820}, 2020.

\end{thebibliography}
\bibliographystyle{IEEEtran}

\end{document}